\documentclass[prx,superscriptaddress,twocolumn,preprintnumbers,%
  amsmath,amssymb,longbibliography]{revtex4-1}
  
\usepackage[breaklinks,colorlinks=true]{hyperref}
\usepackage[T1]{fontenc}
\usepackage[utf8]{inputenc}
\usepackage{mathtools}
\usepackage{amsmath,amsthm,amssymb}
\usepackage{amsmath,bm}
\usepackage{lipsum}
\usepackage{calrsfs}
\DeclareMathAlphabet{\pazocal}{OMS}{zplm}{m}{n}

\usepackage{color}
\allowdisplaybreaks

\usepackage[dvipsnames]{xcolor}
\newcommand\myshade{85}
\colorlet{mylinkcolor}{BrickRed}
\colorlet{mycitecolor}{NavyBlue}
\colorlet{myurlcolor}{Aquamarine}
\hypersetup{
  linkcolor  = mylinkcolor!\myshade!black,
  citecolor  = mycitecolor!\myshade!black,
  urlcolor   = myurlcolor!\myshade!black,
  colorlinks = true,
}

\DeclareMathOperator*{\argmax}{arg\,max}
\DeclareMathOperator*{\argmin}{arg\,min}

\newcommand{\tE}{\tilde{\textbf{E}}}
\newcommand{\tM}{\tilde{\textbf{M}}}

\begin{document}
\title{A generalization of the maximum entropy principle for curved statistical manifolds} 
\author{Pablo A. Morales}
\affiliation{Research Division, Araya Inc., 
Tokyo 107-6019, Japan}
\email{pablo$_$morales@araya.org}
\author{Fernando E. Rosas}
\affiliation{Data Science Institute, Imperial College London, London SW7 2AZ, UK}
\affiliation{Centre for Psychedelic Research, Department of Brain Science, Imperial College London, London SW7 2DD, UK}
\affiliation{Centre for Complexity Science, Imperial College London, London SW7 2AZ, UK}

\newtheorem{definition}{Definition}
\newtheorem{theorem}{Theorem}
\newtheorem{lemma}{Lemma}
\newtheorem{proposition}{Proposition}
\newtheorem{corollary}{Corollary}
\newtheorem{example}{Example}
\newtheorem{remark}{Remark}

\begin{abstract}
The maximum entropy principle (MEP) is one of the most prominent methods to investigate and model complex systems. Despite its popularity, the standard form of the MEP can only generate Boltzmann-Gibbs distributions, which are ill-suited for many scenarios of interest. 
As a principled approach to extend the reach of the MEP, this paper revisits its foundations in information geometry and shows how the geometry of curved statistical manifolds naturally leads to a generalization of the MEP based on the R\'enyi entropy. By establishing a bridge between non-Euclidean geometry and the MEP, our proposal sets a solid foundation for the numerous applications of the R\'enyi entropy, and enables a range of novel methods for complex systems analysis.

\end{abstract}

\maketitle

\section{Introduction}

The progressive unveiling of the intricate connections that exists between information theory and statistical mechanics has allowed fundamental advances on our understanding of complex systems~\cite{thurner2018introduction}. 
One of the most important methods resulting from those discoveries is the \emph{maximum entropy principle} (MEP), which unifies multiple results and procedures under a single heuristic that operationalizes Occam's razor~\cite{jaynes1957information,jaynes2003probability}. From a pragmatic perspective, the MEP can be understood as a modeling framework that is particularly well-suited for building statistical descriptions of a broad class of systems in contexts of incomplete knowledge~\cite{cofre2019comparison}. The high versatility of the MEP has allowed it to find applications in a wide range scenarios, including the analysis of DNA motifs of transcription factor binding sites~\cite{santolini2014general}, co-variations in protein families and amino acid contact prediction~\cite{weigt2009identification,morcos2011direct}, diversity of antibody repertoires in the immune system~\cite{mora2010maximum,elhanati2014quantifying}, coordinated firing patterns of neural populations~\cite{schneidman2006weak,tang2008maximum,marre2009prediction,cofre2014exact}, collective behavior of bird flocks and mice~\cite{bialek2012statistical, cavagna2014dynamical,shemesh2013high}, the abundance and distribution of species in ecological niches~\cite{harte2011maximum,harte2014maximum}, and patterns of behavior in various complex human endeavours~\cite{lee2018partisan,lynn2019surges}.

The efficacy of the MEP rests on Shannon's entropy, which acts as as an estimate of ``uncertainty'' that guides the modeling procedure. Colloquially, the MEP generates the statistical model that is less structured while being consistent with the available knowledge,  building on the available knowledge but nothing else. However, the functional form of the Shannon entropy greatly restricts the range outputs that the MEP can offer. In particular, standard applications of the MEP can only generate Boltzmann-Gibbs distributions, which are unsuitable to describe complex systems displaying long-range correlations or other effects related to different types of statistics~\cite{tsallis2005asymptotically,jizba2004world,thurner2007entropies,tsallis2009introduction}. This important limitation have triggered various efforts to generalize the MEP by means of leveraging generalizations of Shannon's entropy, resulting in a rich array of proposals (see e.g.~\cite{tsallis1988possible,beck2003superstatistics,jizba2004observability,hanel2012generalized,jeldtoft2018group}). However, we argue that plugging a generalized entropy into the MEP framework inevitably leads to an adhoc procedure whose value is fundamentally hindered by the heuristic nature of the MEP itself.

An alternative approach to extend the MEP is to consider it not as a stand-alone principle, but as a consequence of deeper mathematical laws. One route to do this --- that we follow in this paper --- is to regard the MEP as a direct consequence of the 
geometry of statistical manifolds~\cite[Sec.III-D]{amari2001information}. In effect, by leveraging the structure of dual orthogonal projections allowed by the flat geometry associated with the Kullback–Leibler  divergence~\cite{amari2000methods,amari2016information}, the seminal work of Amari established how 
the standard MEP naturally emerges when considering hierarchical ``foliations'' of the manifold. This perspective not only sets the MEP on a firm mathematical bases, but further endows it with sophisticated tools from information geometry --- which can be used e.g. to disentangle the relevance of interactions of different orders within the system~\cite{schneidman2003network,olbrich2015information,rosas2016understanding}.

In this paper we show how the geometry of curved statistical manifolds naturally leads to an extension of the MEP based on the R\'enyi entropy. 
In contrast to flat cases, the geometrical structure of curved statistical manifolds disrupts the standard construction of orthogonal projections based on Legendre-dual coordinates, making the analysis of foliations highly non-trivial. Nonetheless, by leveraging the rich literature on curved statistical manifolds~\cite{amari2000information,Kurose2002,matsuzoe2010statistical,amari2016information,wong2018logarithmic,scarfone2020study}, the framework put forward in this paper reveals how the geometry established by the R\'enyi divergence is suitable for establishing hierarchical foliations that, in turn, lead to a generalization of the MEP.

The results presented in this paper serve to emphasize the special place that the R\'enyi entropy has among other generalized entropies --- at least from the perspective of the MEP. 
Furthermore, it provides a solid mathematical foundation for the plethora of existent applications based on the R\'enyi entropy (see e.g. Refs.~\cite{shalymov2016dynamics,bashkirov2004maximum,GeometricMutInf,DetectingPhaseTwithRenyi}). Furthermore, the novel connection established between information geometry and this generalized MEP opens the door for fertile explorations combining non-Euclidean geometry methods and statistical analyses, which may lead to new insights and techniques to further deepen our understanding of complex systems.

The rest of this article is structured as follows. First, Section~\ref{sec:alpha} provides a brief introduction to information geometry, emphasising concepts that are key to our proposal. Then, Section~\ref{sec:hiearchical} develops the analysis of foliations in curved statistical manifolds, and Section~\ref{sec:maxentRenyi} establishes its relationship with a maximum R\'enyi entropy principle. Finally, Section~\ref{sec:discussion} discusses the implications of our findings and summarizes our main conclusions.

\section{Preliminaries}
\label{sec:alpha}

\subsection{The Dual Structure of Statistical Manifolds}
\label{sec:prel1}

Our exposition is focused on statistical manifolds $\mathcal{M}$, whose elements are probability distributions $p_\xi(x)$ with $x\in\boldsymbol\chi$ and $\xi\in\mathbb{R}^d$. 
The geometry of such statistical manifolds 
is determined by two structures: a metric tensor $g_p$, 
and a torsion-free affine connection 
pair $(\nabla, \nabla^{*})$ that are dual with respect to $g_p$. Intuitively, $g_p$ defines norms and angles between tangent vectors and, in turn, establishes curve length and the \emph{shortest} curves. 
On the other hand, the affine connection establishes contravariant derivatives of vector fields establishing the notion of parallel transportation between neighbouring tangent spaces, which defines what is a \emph{straight} curve.

Traditional Riemannian geometry is build on the assumption that the shortest and the straightest curves coincide, which led to the study of metric-compatible (Levi-Civita) connections ---  pivotal to the development of the theory of general relativity. However, modern approaches motivated in information geometry~\cite{amari2021information} and 
gravitational theories~\cite{Vitagliano:2010sr,Vitagliano:2013rna} consider more general cases, where the metric and connections are independent from one another. In such geometries, the parallel transport operator $\Pi:T_p\mathcal{M}\to T_q\mathcal{M}$ and its dual $\Pi^*$~\footnote{The dual transport operator acts on cotangent vectors, and is defined by the condition of guaranteeing
$g_q (\Pi V,\Pi^{*} W)=g_p (V,W)$ for all $W\in T_p\mathcal{M}$ and $V\in T^*_p\mathcal{M}$).} (induced by $\nabla$ and $\nabla^*$, respectively) might differ. 
The departure of $\nabla$ and $\nabla^*$ from self-duality can be shown to be proportional to Chentsov's tensor, which allows for a single degree of freedom traditionally denoted by $\alpha \in \mathbb{R}$~\cite{amari2021information}. Put simply, $\alpha$ captures the degree of asymmetry between short and straight curves, with $\alpha=0$ corresponding to metric-compatible connections where $\nabla=\nabla^*$.

An important property of the geometry of a statistical manifold ($\mathcal{M},g,\nabla, \nabla^{*}$) is its curvature, which can be of two types: the (Riemann-Christoffel) metric curvature, or the curvature associated to the connection. Both quantities capture the distortion induced by parallel transport over closed curves --- the former with respect to the Levi-Civita connection, and the latter with respect to $\nabla$ and $\nabla^*$. In the sequel we use the term ``curvature'' to refer exclusively to the latter type.

\subsection{Establishing geometric structures via divergences} \label{sub:estab_gvd}

A convenient way to establish a geometry on a statistical manifold is via \emph{divergence maps}~\cite{amari2010information}. 
Divergences are smooth,
distance-like mappings for the form $\pazocal{D}:\mathcal{M}\times \mathcal{M} \to \mathbb{R}$, 
which satisfy $\pazocal{D}(p||q)\ge 0$ and vanish only when $p=q$~\footnote{Divergences are in general weaker than distances, as they don't need to be symmetric in their arguments and don't need to respect the triangle inequality.}. We use the shorthand notation $\pazocal{D}[\xi;\xi'] := \pazocal{D}(p_{\xi}||q_{\xi'})$ when expressing $\pazocal{D}$ under a parametrization of $\mathcal{M}$ in terms of coordinates $\xi=(\xi^1,\dots,\xi^n)$~\cite{amari2001information}; divergences in this form are often called ``contrast functions'' (see Ref.~\cite[Sec.~11]{calin2014geometric}). 

Let us see how one can naturally build a metric from a contrast function~\cite[Sec.~4]{amari2010information}. A metric $g(\xi)$ can be built from the second-order expansion of the divergence $\pazocal{D}$ as
\begin{equation}
\label{eq:FisherMetric}
  g_{ij}(\xi) 
  = \left \langle \partial_{\xi^i},
  \partial_{\xi^j} \right \rangle 
  = - \partial_{\xi^i , \xi^{\prime j}}
  \pazocal{D}[\xi;\xi'] \big|_{\xi=\xi'}~,
\end{equation}
which is positive-definite due to the non-negativity of $\pazocal{D}$. 
This construction leads to the \emph{Fisher's metric}, which is the unique metric that emerges from a broad class of divergences~\cite[Th.~5]{amari2010information}, with this being this closely related Chentsov’s theorem~\cite{chentsov1982statistical,ay2015information,van2017uniqueness,dowty2018chentsov}.
Analogously, connections (or equivalently Christoffel symbols) emerge at the third order expansion of the
divergence as follows:
\begin{subequations}
\label{eq:connections}
\begin{align}
\label{eq:C1}
  \Gamma_{ijk}(\xi) & = \left \langle \nabla_{\partial_{\xi^i}}
  \partial_{\xi^j} ,\partial_{\xi^k} \right \rangle 
  = -\;\left . \partial_{i ,j} \partial_{k'} \pazocal{D}[\xi;\xi'] \right|_{\xi=\xi'}\!,\\
\label{eq:C2}
  \Gamma_{ijk}^{*}(\xi) & = \left \langle \nabla_{\partial_{\xi^i}}^{*}
  \partial_{\xi^j} ,\partial_{\xi^k} \right \rangle 
  = -\left . \partial_{k} \partial_{i',j'} \pazocal{D}[\xi;\xi'] \right|_{\xi=\xi'}\!,
\end{align}
\end{subequations}
where the shorthand notation $\partial_{\xi^i} =\partial_{i}$ and $\partial_{\xi'^i} =\partial_{i'}$ has been adopted for brevity.
In summary, Fisher's metric is insensible the choice of divergence but the resulting connections are,
and therefore the effects of a particular $\pazocal{D}$ manifest only at third-order. 
Interestingly, this construction relating the metric and connections with
the second and third derivatives of a scalar potential bears a striking resemblance to K{\"a}hler structures on complex manifolds, which can be built through further constraints and are applicable to a range of inference problems~\cite{Choi_2015,zhang2013symplectic}.

The approach of building geometries based on divergences does not lack generality, as it has been shown that any geometry can be expressed by an appropriate divergence~\cite{matumoto1993any,ay2015novel}. 
Of the various types of divergences explored in the literature (c.f. \cite{Liese2006Divergences} and references within), two classes are particularly important: \emph{$f$-divergences} of the form
\begin{equation}
      \pazocal{D}_{f}[\xi;\xi'] = \int_{\boldsymbol\chi} p_{\xi}(x) f\left( \frac{p_{\xi}(x)}{q_{\xi'}(x)} \right)
    d\mu(x)
\end{equation}
for $f(x)$ convex with $f(1)=0$, and \emph{Bregman divergences} of the form
\begin{align}
      \pazocal{D}_{\phi}[\xi;\xi'] 
      &= (\xi-\xi')\cdot \mathrm{D}\phi(\xi') - \big(\phi(\xi) - \phi(\xi')\big) \\
      &= \xi\cdot\eta' - \phi(\xi) - \psi(\eta')\label{eq:legendre}
\end{align}
for $\phi(\xi)$ a concave function~\footnote{Following Ref.~\cite{wong2018logarithmic}, we use a non-standard definition of Bregman divergences based on concave (instead of convex) functions.}, with $\mathrm{D}\phi=(\partial\phi/\partial\xi_1,\dots,\partial\phi/\partial\xi_d)$ denoting the gradient of $\phi$, $\psi(\eta)=\min_\xi\big(\eta\cdot\xi-\phi(\xi)\big)$ is the Fenchel–Legendre concave conjugate of $\phi$, and $\eta$ the dual coordinates of $\xi$ such that
\begin{equation} \label{eq:scalar_pot}
    \xi = \mathrm{D}\psi(\eta) 
    \quad\text{and}\quad
    \eta = \mathrm{D}\phi(\xi)~.
\end{equation}
Each of these types of divergences have important properties from an information geometry perspective: $f$-divergences are monotonic with respect to coarse-grainings of the domain of events $\boldsymbol\chi$, while Bregman divergences enable dual structures that set the basis for orthogonal projections~\cite{Amaridivergences}. 

As mentioned above, the deviation of a given connection $\nabla$ from its corresponding metric-compatible (i.e. Levi-Civita) counterpart can be measured by $\alpha T$, where $T$ corresponds to the invariant \textit{Amari-Chensov} tensor~\cite{cencov2000statistical,amari1982differential} and $\alpha \in \mathbb{R}$ is a free parameter. The invariance of $T$ implies that the value of $\alpha$ entirely determines the connection, and the corresponding geometry can be obtained from a divergence of the form
\begin{align}
\pazocal{D}_{\alpha}(p || q)&= \frac{4}{1-\alpha^2} \int_{\boldsymbol\chi}
  \left \{ 1 - p^{\frac{1-\alpha}{2}}q^{\frac{1+\alpha}{2}} \right \}d\mu(x)~,
\end{align}
which is known as \emph{$\alpha$-divergence}. As important particular cases, if $\alpha=0$ then $\pazocal{D}_{\alpha}$ becomes the square of Hellinger's distance, and if $\alpha=1$ then it gives the well-known Kullback-Leibler
\begin{equation}
  \pazocal{D}_{\mathrm{KL}}(p||q) = \int_{\boldsymbol\chi} p(x) \log \left( \frac{p(x)}{q(x)} \right)
  d\mu(x)~.
\end{equation}

It is worth noting that geometrical structures are invariant under certain types of transformations. 
For example, consider a divergence $\tilde{\pazocal{D}}$ given by  $\tilde{\pazocal{D}}[\xi;\xi']:=F(\pazocal{D}[\xi;\xi'])$, with $F$ a monotone and differentiable function satisfying $F(0)=0$
~\footnote{Given two divergences $\pazocal{D}$ and $\tilde{\pazocal{D}}$, there is always a function $F:\mathbb{R}^3\to\mathbb{R}$ such that 
$\tilde{\pazocal{D}}[\xi;\xi']=F(\pazocal{D}[\xi;\xi'],\xi,\xi)$. Building on this fact, one can consider three levels of similarity: (i) when $F$ depends only on the first argument --- which then implies the corresponding geometries are essentially the same, (ii) when $F$ can be expressed as $F(x,y,z) = f(x) g(y,z)$ --- which implies conformal-projective equivalence (see Sec.~\ref{sec:proj}, and (iii) the more general case.}. 
Then, it can be shown using Eqs.~\eqref{eq:FisherMetric} and \eqref{eq:connections} that the metric and connections induced by $\pazocal{D}$ and $\tilde{\pazocal{D}}$ are related as follows:
\begin{equation}
\label{eq:class_def}
\tilde{g}=F'(0)\, g,
\quad
\tilde{\Gamma} = F'(0)\,\Gamma,
\quad 
\tilde{\Gamma}^{*} = F'(0)\,\Gamma^{*}~.
\end{equation}
Therefore, $\tilde{\pazocal{D}}$ gives rise to exactly the same geometrical structure when $F'(0)=1$, and a scaled version otherwise. More general transformations between divergences and their corresponding geometries are discussed in Section~\ref{sec:proj}.

\subsection{A Pythagorean relationship in curved spaces via the R\'enyi divergence}\label{sec:IIc}

The connection induced by the KL divergence and its natural coordinates is flat (i.e. $\Gamma_{ijk}(\xi)=\Gamma^*_{ijk}(\xi)=0$). 
However, this does not hold for $\alpha$-divergences when $\alpha \neq1$, which retain the same Fisher's metric but induce a connection with constant sectional curvature $\omega=(1-\alpha^2)/4$ over the whole manifold~\cite[Theorem~7]{wong2018logarithmic}. 
This results into a spherical ($S^n$) geometry for $\alpha \in (0,1)$, or an hyperbolic ($H^n$) geometry for $\alpha \notin (0,1)$.

A non-zero curvature affects the relationship between  geodesics~\footnote{Geodesics are the straight curves established by the connection, which in non-Riemannian geometries are not the same as the shortest curves between two points.}: if the ``$\alpha$-geodesic'' joining $p$ and $q$ is orthogonal (with respect to the Fisher metric) to the one 
joining $q$ and from $r$, then
\begin{align}
\label{eq:Add_prop}
  \pazocal{D}_{\alpha}(p || r) = & \, \pazocal{D}_{\alpha}(p || q) + \pazocal{D}_{\alpha}(q || r) \nonumber \\  & - \frac{1-\alpha^2}{4}\pazocal{D}_{\alpha}(p || q)\pazocal{D}_{\alpha}(q || r)~,
\end{align}
resulting in a deviation from the standard ``Pythagorean relationship''
that is observed for the case of $\alpha=1$~\cite{amari2000methods}. 
However, one can rewrite Eq.~(\ref{eq:Add_prop}) as
\begin{equation}
\label{eq:divsphere}
  1-\omega \pazocal{D}_{\alpha}(p || r) = \big(1-\omega \pazocal{D}_{\alpha}(p || q)\big) 
  \big(1-\omega \pazocal{D}_{\alpha}(q || r)\big),
\end{equation}
which describes the relationship between angles on the sphere or hyperbolic space --- depending on the sign of $\omega$~\cite{amari2000methods}. 
Interestingly, Eq.~\eqref{eq:divsphere} suggests that a divergence of the form
\begin{align} 
\label{eq:Renyialpha}
    \mathcal{D}_{\gamma}(p || q) 
    :=& \frac{1}{\gamma} \log\big(1 + \gamma(1+\gamma)\pazocal{D}_{\alpha} (p||q)\big) \\
    =& \frac{1}{\gamma} \log\int_{\boldsymbol\chi} p(x)^{\gamma + 1}q(x)^{-\gamma} d\mu(x)  \label{eq:Renyi_div}
\end{align}
with $\alpha=-1-2\gamma$ would recover the ``Pythagorean relationship.'' 
In fact, $\mathcal{D}_{\gamma}$ can be recognized as the well-known R\'enyi divergence of order $\gamma-1$~\cite{wong2018logarithmic,amari2021information}, noting that we follow Ref.~\cite{valverde2019case} in adopting a shifted indexing.

The R\'{e}nyi divergence is an $f$-divergence with $f(x)=x^{\gamma}$ but it is not a Bregman divergence; however, one can re-cast it as a ``Bregman-like'' divergence~\cite{wong2018logarithmic}. To see this, let's consider $\tilde{p}_\xi \in\mathcal{M}$ to be a deformed exponential family distribution of the form (see Appendix~\ref{app:def_exp_dist})
\begin{equation}
\label{eq:def_exp}
  \tilde{p}_\xi(x) = \big(1+ \gamma \xi \cdot h(x)\big)^{-\frac{1}{\gamma}} e^{\varphi_{\gamma}(\xi)}~,
\end{equation}
where $h(x)\in\mathbb{R}^d$ is a vector of sufficient statistics of $x$ and $\varphi_{\gamma}$ is a normalising potential given by
\begin{equation}\label{eq:def_varphi}
  \varphi_{\gamma}(\xi) := -\log \int_{\boldsymbol \chi} (1+\gamma \xi \cdot h(x))^{-\frac{1}{\gamma}}d\mu(x)~.
\end{equation}
Note that Eq.~\eqref{eq:def_exp} gives a standard exponential family distribution when $\gamma\to0$. By defining $\mathcal{D}_\gamma[\xi;\xi']  :=\mathcal{D}_\gamma(\tilde{p}_\xi || \tilde{p}_{\xi'})$ to be the corresponding contrast function of the R\'enyi divergence, then one can show that~\cite[Th.13]{wong2018logarithmic}
\begin{equation}\label{eq:def_bigd}
    \mathcal{D}_\gamma[\xi;\xi'] 
    = \frac{1}{\gamma} \log(1 + \gamma \xi\cdot \eta^{\prime})
    - \varphi_{\gamma}(\xi) - \psi_{\gamma}(\eta')~,
\end{equation}
which resembles Eq.~\eqref{eq:legendre} but with the factor $\xi \cdot \eta$ replaced by a logarithm. Above, 
\begin{equation}
    \psi_{\gamma}(\eta) :=
    \min_{\xi} \Big\{ \frac{1}{\gamma} \log(1+\gamma \xi\cdot \eta) - \varphi_{\gamma}(\xi)\Big\}
\end{equation}
is a generalization of the Fenchel–Legendre transform of $\varphi_{\gamma}$, which has conjugate coordinates established by
\begin{subequations}
\begin{align}
    \eta &= \frac{1}{1+\gamma \xi\cdot\mathrm{D}\varphi_{\gamma}(\xi)} \mathrm{D}\varphi_{\gamma}(\xi)
    ~,\label{eq:dual_var0}\\
    \xi &= \frac{1}{1+\gamma \xi \cdot \mathrm{D}\psi_{\gamma}(\eta)} \mathrm{D}\psi_{\gamma}(\eta)~,
    \label{eq:dual_var1}
\end{align}
\end{subequations}
with $\mathrm{D}\varphi$ denoting the Euclidean gradient of $\varphi$.
Finally, it is worth noting that
\begin{equation}
\label{eq:grad_EV}
    \mathrm{D}\varphi_{\gamma} (\xi) = \mathbb{E}_{\xi}\left\{\frac{h(X)}{1+\gamma \xi \cdot h(X)} \right\}=:\mathbb{E}_{\xi}\{Z_{\xi}(h)\}~,
\end{equation}
where $X$ is a random variable that follows the distribution $p_\xi(x)$, $h(X)$ denotes the sufficient statistics of $X$, and $Z_{\xi}(h)$ is defined implicitly as the quantity within the curly brackets. Hence these generalized Fenchel-Legendre dual coordinates can be alternatively expressed as
\begin{equation}\label{eq:dual_esp}
    \eta =  \frac{1}{1+\gamma \xi\cdot \mathbb{E}_{\xi}\{Z_{\xi}(h)\}} \mathbb{E}_{\xi}\{Z_{\xi}(h)\}~.
\end{equation}
For the case of $\gamma=0$, Eq.~\eqref{eq:dual_esp} reduces to the well-known relationship given by $\eta = \mathbb{E}_{\xi}\{h(X)\}$, (see Appendix~\ref{app:special_cases} for further comments).

\begin{figure}
  \includegraphics[width=8.5cm,height=6cm,keepaspectratio]{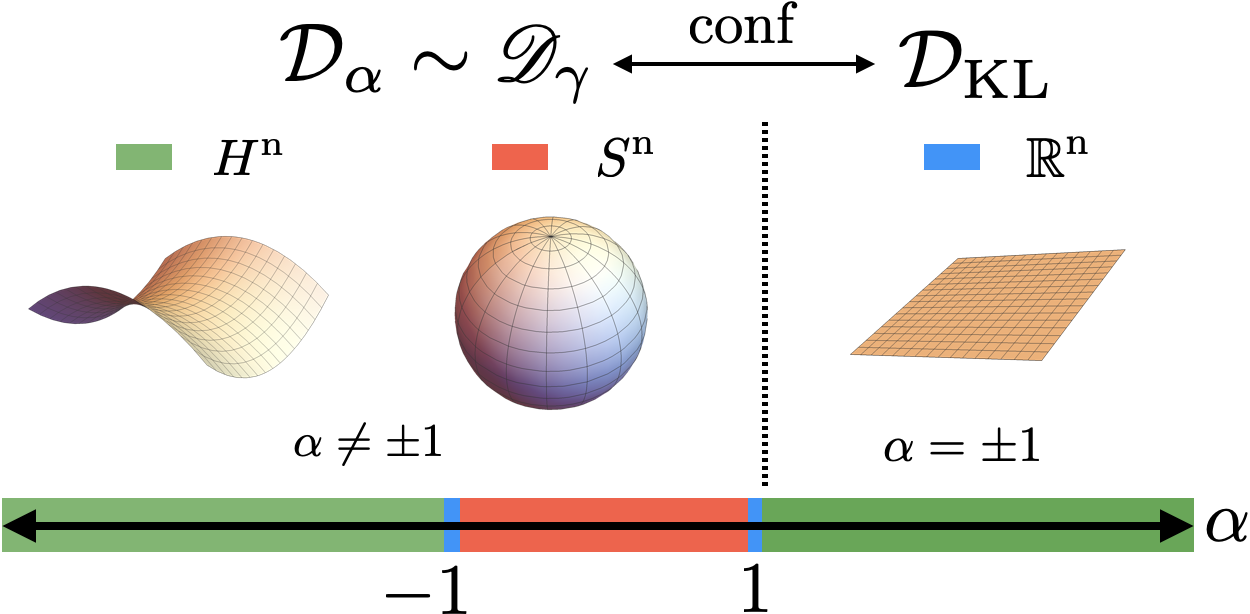}
  \caption{An schematic diagram depicting the three classes of geometrical structures that arise from their $\alpha$-value. The curved (i.e. $\alpha \neq \pm1$) geometries are characterized by the $\alpha$- and R\'enyi's divergence, both of which are conformally-projectively related to the $KL$ divergence --- which in turn generates a flat geometry.}
\end{figure}

\subsection{Conformal-projective classes}\label{sec:proj}

Conformal transformations are operations over geometric structures that are angle-preserving, amounting to (pseudo) rotations and dilation of the points in the manifold. Technically, a conformal transformation on $\mathcal{M}$ is defined as an invertible map $\omega:\mathcal{M}\to \mathcal{M}$
such that the induced metric by the pull-back map $\omega_{*}:T_{\omega(p)}\mathcal{M}\to T_p\mathcal{M}$ is related to the original metric up to a scaling factor $\lambda(p):\mathcal{M}\to\mathbb{R}$ such that 
\begin{equation}
    \label{eq:conformal_metric}
    g_{p}(\omega_{*}(X),\omega_{*}(Y))=\lambda(\omega(p))g_{\omega(p)}(X,Y)
\end{equation}
for all $X,Y\in T_{\omega(p)}\mathcal{M}$. 
Correspondingly, two metrics $g$ and $\tilde{g}$ are said to be \emph{conformally equivalent} if they can be linked via a conformal factor $\lambda$ as in Eq.~\eqref{eq:conformal_metric}.

Due to their non-Riemannian geometry, geometric transformations on statistical manifolds that are ``structure-preserving'' are not fully specified by their effect on the metric, but also need to characterize its effect on the connections --- which may diverge from metric-dependence via Chentsov's tensor. This characterization can be done by relaying on the notion of \emph{projectively equivalence}: two connections $\nabla$ and $\tilde{\nabla}$ are said to be \emph{projectively equivalent} if there exists a 1-form $\nu = a_{i}(\xi) \text{d}\xi^i$ that satisfies
\begin{equation}
  \Gamma^{k}_{ij}(\xi) = \tilde{\Gamma}^{k}_{ij}(\xi) + a_{i}(\xi)\delta^{k}_{j} 
  + a_{j}(\xi)\delta^{k}_{i}~,
\end{equation}
with $\delta_i^j$ the Kronecker delta~\footnote{Equivalently, projective equivalence can be defined by $\nabla_X Y = \tilde{\nabla}_X Y + \nu(X)Y + \nu(Y)X$ for any smooth pair of vector fields $X$ and $Y$.}.

A convenient way to put these notions together and build conformal-projective transformations is by considering transformations over divergences. 
Two divergences $\pazocal{D}$ and $\tilde{\pazocal{D}}$ are said to belong to the same \emph{conformal-projective class} if two conditions are met: (i) their induced metrics are conformally equivalent, and (ii) their induced connections are projectively equivalent. It can be shown that two divergences belong to the same conformally-projective class if and only if they satisfy
\begin{equation}
    \tilde{\pazocal{D}}[\xi;\xi']=\lambda(\xi)\pazocal{D}[\xi;\xi']~,
\end{equation}
with $\lambda$ being the \emph{conformal-projective factor}~\footnote{In general $\lambda$ could depend on both $\xi$ and $\xi'$~\cite{Kurose2002,amari2021information}; however, for the purposes of this paper we restrict ourselves to consider only ``left conformal-projective factors'' (i.e. $\lambda(\xi)$).}.

Let us now study the relationship between the geometries induced by  $\mathcal{D}_{\gamma}$, $\pazocal{D}_{\alpha},$ and $\pazocal{D}_{\text{KL}}$. 
By considering the inverse of Eq.~\eqref{eq:Renyialpha}, one finds that the function 
\begin{equation}\label{eq:F}
    F(x) = \frac{e^{\gamma x} -1}{(1+\gamma)\gamma}~,
\end{equation}
establishes the diffeomorphism  $F(\mathcal{D}_{\gamma}[\xi;\xi'])= \pazocal{D}_{\alpha}[\xi;\xi']$, which reveals that the R\'enyi divergence and $\alpha$-divergences 
generate exactly the same geometry (as 
described by Eqs.~\eqref{eq:class_def}). 
Building on this fact, and leveraging the Legendre-like form of the R\'enyi entropy shown in Eq.~(\ref{eq:def_bigd}), a direct calculation shows that the action of $F$ on $\mathcal{D}_{\gamma}$ can be expressed as a Bregman divergence $\pazocal{D}_{\phi}$ scaled by a conformal-projective factor~\cite[Th.~1]{wong2019logarithmic}:
\begin{equation}
\label{eq:conformal}
    F(\mathcal{D}_{\gamma}[\xi;\xi'])
    = \kappa(\xi)\pazocal{D}_{\phi}[\xi;\xi']~.
\end{equation}
Above, $\phi$ is a scalar potential given by $\phi(\xi)= e^{\gamma \varphi_0(\xi)}$ with $\varphi_0(\xi)$ as given in Eq.~\eqref{eq:def_varphi}, and the conformal-projective factor $\kappa$ has the form
\begin{equation}\label{eq:kappa}
\kappa(\xi)=-\frac{1}{\gamma\phi(\xi)}~.
\end{equation}
Moreover, please note that  $\pazocal{D}_{\phi}$ describes a dually-flat geometry, belonging to the same equivalent class as the KL divergence. 
Thus, these results together establishes that R\'enyi's $\mathcal{D}_{\gamma}$, $\pazocal{D}_{\alpha},$ and $\pazocal{D}_{\text{KL}}$ belong to the same conformal-projective equivalence class.

To conclude, let us present a derivation of the functional form of $\kappa(\xi)$ used in Eq.~\eqref{eq:conformal} following Ref.~\cite{wong2019logarithmic}. The metric induced by $\mathcal{D}_{\gamma}[\xi;\xi']$ is given by
\begin{align}\label{eq:metrixx}
  \tilde{g}_{ij}(\xi) 
  := -\left . 
  \partial_{i , j\prime}
  \mathcal{D}_{\gamma}[\xi;\xi'] \right|_{\xi=\xi'} 
  = \kappa (\xi)\partial_{ij}\phi(\xi)~,
\end{align}
and hence $\tilde{g}_{ij}(\xi)=\kappa(\xi) g_{ij}(\xi)$. Furthermore, its induced connection and metric curvature can be found to be
\begin{subequations}
\begin{align}
    \label{eq:kappaGamma}
    \tilde{\Gamma}_{ij}^{\,\, k}(\xi) &= \frac{\partial_{i}\kappa(\xi)}{\kappa(\xi)} \delta_{j}^{k} + \frac{\partial_{j}\kappa(\xi)}{\kappa(\xi)} \delta_{i}^{k}~, \\
    \label{eq:kappaRiemann}
    \tilde{R}_{ijk}^{\quad \, l}(\xi) &= \kappa(\xi) \left( \partial_{jk}\frac{1}{\kappa(\xi)} \delta_{i}^{l} - \partial_{ik}\frac{1}{\kappa(\xi)} \delta_{j}^{l} \right)~.
\end{align}
\end{subequations}
Hence, by introducing the 1-form $\nu = \mathrm{d}\log \kappa(\xi)$, 
one can identify the affine connection induced by $\tilde{\Gamma}_{ij}^{\,\, k}(\xi)$ as being projectively flat. This 1-form --- or equivalently, the conformal factor $\kappa(\xi)$ --- can be derived from the
Riemann curvature tensor, which for spaces of constant sectional curvature takes the form $R_{ijk}^{\quad \, l}=K (g_{jk}\delta_{i}^{l}-g_{ik}\delta_{j}^{l})$, with $K \in \mathbb{R}$ corresponding to its scalar curvature. As mentioned in Section~\ref{sec:IIc}, the geometry induced by the $\alpha$-divergence has curvature $\omega=(1-\alpha^2)/4$ throughout the whole manifold, and hence its Riemann tensor can be rewritten as 
\begin{equation}\label{eq:sss}
    R_{ijk}^{\quad \, l}=\frac{1+\alpha}{2} (\tilde{g}_{jk}\delta_{i}^{l}-\tilde{g}_{ik}\delta_{j}^{l})~,
\end{equation} 
where a factor $\frac{1-\alpha}{2}=\gamma +1$ from $\omega$ has been absorbed by the metric~\footnote{Note that the metrics coming from the $\alpha$ and R\'enyi divergences are conformally related $\tilde{g}=(\gamma + 1)g$ as seen by~\eqref{eq:Renyialpha}}. Moreover, using the fact that the Riemann tensor is left unchanged by the conformal-projective transformation (i.e. $\tilde{R}_{ijk}^{\quad \, l} = R_{ijk}^{\quad \, l}$), and recognising that $K=-\gamma$, one can use Eqs.~\eqref{eq:metrixx}, \eqref{eq:kappaRiemann} and \eqref{eq:sss} to show that
\begin{equation}\label{eq:sads}
    \frac{1}{\kappa(\xi)} = -\gamma \phi(\xi) + \sum_i a_{i}\xi^{i} + b~,
\end{equation}
for some $a_{i},b \in \mathbb{R}$. Finally, as the linear terms can be absorbed in the definition of $\phi$, Eq.~\eqref{eq:sads} leads to the expression for $\kappa(\xi)$ as shown above.

\section{Orthogonal foliations in curved statistical manifolds}
\label{sec:hiearchical}

This section presents the study of orthogonal foliations in curved statistical manifolds. For simplicity of the exposition, the rest of the paper focuses on multivariate distributions of $n$ binary random variables --- i.e. distributions of the form $p(x)$ where $x=(x_1,\dots,x_n)$ with $x_i\in\{0,1\}$, and hence $\boldsymbol\chi = \{0,1\}^n$.

\subsection{Orthogonal foliations on flat-projective spaces}
\label{sec:dual_foliation}

Let us consider a parametrization $\nu$ of the manifold $\mathcal{M}$. Then, for a given $p_\nu\in\mathcal{M}$ we define 
\begin{equation}\label{eq:M}
    \tM_k\{p_\nu\} := \{ q_{\nu'}\in\mathcal{M} | \nu_i' = \nu_i \;\forall i=1,\dots,k\}~,
\end{equation}
which establishes a nested structure on the manifold of the form
\begin{align}
\label{eq:foliationM}
  \{p\} = \tM_n\{p\}
  \subset \tM_{n-1}\{p\} 
  \subset \dots 
  \subset \tM_{0}\{p\} = \mathcal{M}~.
\end{align}
The parametrization $p_\nu$ also induces a natural basis for the cotangent space at each $p\in\mathcal{M}$, which we denote by $\partial_{\nu_i}\in T^*_{p}\mathcal{M}$. To study the geometry of this basis, let's consider
the functional form of $\mathcal{D}_\gamma$ induced by $\nu$, which is given by $\mathcal{D}_\gamma[\nu ; \nu'] := \mathcal{D}_\gamma( p_\nu || p_{\nu'} )$. Then, the inner product between the basis elements $\partial_{\nu_i}$ can be calculated as
\begin{equation}
    \langle \partial_{\nu_i}, \partial_{\nu'_j} \rangle 
    = -\partial_{\nu_i,\nu^\prime_j} \mathcal{D}_\gamma[\nu;\nu']\big|_{\nu'=\nu} = \tilde{g}^{ij}(\nu)~.
\end{equation}

The properties of $\mathcal{D}_\gamma$ guarantees that $\tilde{g}^{ij}(\nu)$ is positive-definite, and hence it has a well-defined inverse for each $\nu$ which we denote by $r^{ij}(\nu) := \big(g^{-1}(\nu)\big)^{ij}$.
By denoting as $\theta$ the primal coordinates  
with respect to $r$, one can then define
\begin{equation}\label{eq:E}
\tE_k := \{ p_{\theta}\in\mathcal{M} | \theta_j=\theta_j^u,\;\forall j>k\},     
\end{equation}
where $\theta^u$ denote the $\theta$-coordinates of the uniform distribution $u$. It is direct to verify that 
\begin{align}
\label{eq:foliation}
  \{u\}=\tE_0 \subset \tE_1 \subset \dots \subset \tE_n = \mathcal{M}~.
\end{align}
Interestingly, $\tE_k$ grows with $k$ while $\tM_k$ shrinks such that for each $k$ their combined dimensions sum up to $n$ --- being enough to account for the dimensionality of  $\mathcal{M}$. Furthermore, due to the fact that these complementary dimensions are orthogonal, this implies that their intersection cannot be empty.

We summarize these ideas in the following definition.
\begin{definition}
For a given parametrization $\nu$ of $\mathcal{M}$ for which $\tE_k$ exists, then the \emph{orthogonal foliation} of $\mathcal{M}$ associated to $p_\nu$ is the collection of sets $\big\{\tM_k\{p_\nu\},\tE_k\big\}$.
\end{definition}

Please note that the bases of $T_{p}\mathcal{M}$ and $T^*_{p}\mathcal{M}$ determined by the generalized Fenchel-Legendre dual coordinates established by Eqs.~\eqref{eq:dual_var0} and \eqref{eq:dual_var1} are 
not orthogonal under the inner product related to the scalar potential $\varphi$ and its conjugate if $\gamma>0$, as discussed in Appendix~\ref{app:pythagorean}. Therefore, the standard relationship between geometric duality and Fenchel-Legendre duality that holds for $\gamma=0$ is broken in curved statistical manifolds. 
Nonetheless, projective-flatness allows for the metric induced by $\mathcal{D}_\gamma$ to be expressible in coordinates where the bases are manifestly orthogonal up to a conformal-projective factor, so that $\langle\partial_{\xi_i},\partial_{\eta^j}\rangle = \kappa(\theta) \delta_{i}^{j}$ with $\kappa(\theta)$ as defined in Eq.~\eqref{eq:kappa}. Then, $\theta$ and its Fenchel-Legendre conjugate established by Eq.~\eqref{eq:scalar_pot} define
a set of conformal-projective coordinates.

Crucially, orthogonal foliations satisfy a Pythagorean property, as shown by the following lemma.

\begin{lemma}\label{lemma:orto}
Given an orthogonal foliation $\{\tM_k\{p\},\tE_k\}$, 
if $p\in \tM_k\{p\}$, $r\in \tE_k$, and $q \in \tM_k\{p\} \cap \tE_k$ then
\begin{equation}\label{eq:Pyth_rel}
    \mathcal{D}_{\gamma}(p||r) 
    = \mathcal{D}_{\gamma}(p||q) 
    + \mathcal{D}_{\gamma}(q||r)~.
\end{equation}
\end{lemma}
\begin{proof}
See Appendix~\ref{app:pythagorean}.
\end{proof}

It is important to note that while building orthogonal coordinates is a relatively simple construction, these don't necessarily generally guarantee a Pythagorean relationship. As a matter of fact, although the equivalence between R\'enyi's and $\alpha$-divergences 
ensures that both divergences induce the same geometry, only R\'enyi's exhibits a correspondence between orthogonality on the metric and a Pythagorean relationship on the divergence (see Section~\ref{sec:IIc}).
To illustrate these ideas, let us consider a particular construction where we take $\tM_k$ as the set of probabilities distributions with fixed expectation values, denoted by $\eta$, and come up with its orthogonal complement. From $\phi$ as the potential encoding these change of coordinates, we define its conjugate potential $\bar{\psi}= \min_{\xi}(\xi\cdot \eta -\phi(\xi))$. In this way, the primal coordinates $\bar{\xi}$ orthogonal to $\eta$ follow from $\mathrm{D}(\xi\cdot \eta -\phi(\xi))$, that is,
\begin{equation}
    \bar{\xi}^i = \mathbb{E}_{\xi}\{h^i (x)\} -\frac{1}{\gamma \kappa(\xi)}(\mathrm{D}\log \kappa(\xi))^i~,
\end{equation}
where the first term in the right hand side follows from $\eta^i = \mathbb{E}_{\xi}\{h^i (x)\}$. The primal coordinates $\bar{\xi}^i$, allows to construct an orthogonal complement to $\tM_k$, and from \eqref{eq:exp_fam} one finds that
\begin{equation}
    \bar{\textbf{E}}_k(c_{k^+}) 
    =
    \{ p_{\bar{\xi}}(x) \in \mathcal{M} \:|\: \bar{\xi}_{k^+} = c_{k^+} \}.
\end{equation}

\subsection{Higher-order hierarchical decomposition} 

Using a orthogonal foliation, we now introduce the notion of hierarchical decomposition on curved statistical manifolds.

\begin{definition}
The $k$-th order $\gamma$-projection of $p\in\mathcal{M}$ under the orthogonal foliation $\{\tM_k\{p\},\tE_k\}$ is
\begin{equation}\label{eq:def_projection}
  \tilde{p}^{(k)} := \argmin_{q \in \tE_k} \mathcal{D}_{\gamma}(p\,;q)
  = \argmin_{q \in \tE_k} \pazocal{D}_{\alpha}(p\,;q)~.
\end{equation}
\end{definition}
Above, the minimum under $\mathcal{D}_{\gamma}$ and $\pazocal{D}_{\alpha}$ is the same, as both divergences are related by a monotonous function as shows by Eq.~\eqref{eq:Renyialpha}. 
An useful property of the orthogonal foliation is that it enables a useful characterization of $\tilde{p}^{(k)}$ for $k>0$, as shown in the next Lemma.

\begin{lemma}\label{lemma:inc}
The $k$-th order $\gamma$-projection of $p\in\mathcal{M}$ satisfies 
$\{\tilde{p}^{(k)}\} = \tE_k \cap \tM_{k}\{p\}$.
\end{lemma}

\begin{proof}
Consider $q\in\tE_k \cap \tM_{k}\{p\}$. It is direct to verify that $p,q\in\tM_k\{p\}$ 
and $q,\tilde{p}^{(k)}\in\tE_k$. 
Then, Lemma~\ref{lemma:orto} 
implies that
\begin{equation}\label{eq:asdvbg}
    \mathcal{D}_{\gamma}(p||\tilde{p}^{(k)}) 
    = \mathcal{D}_{\gamma}(p||q) 
    + \mathcal{D}_{\gamma}(q||\tilde{p}^{(k)}) \geq \mathcal{D}_{\gamma}(p||q) ~.
\end{equation}
Additionally, Eq.~\eqref{eq:def_projection} and the fact that $q\in\tE_k$ imply that $\mathcal{D}_{\gamma}(p||q) \geq \mathcal{D}_{\gamma}(p||\tilde{p}^{(k)})$, which together with Eq.~\eqref{eq:asdvbg} show that  $\mathcal{D}_{\gamma}(p||\tilde{p}^{(k)})=\mathcal{D}_{\gamma}(p||q)$. This, combined again with Eq.~\eqref{eq:asdvbg}, implies in turn that $\mathcal{D}_{\gamma}(q||\tilde{p}^{(k)})=0$, which can only be satisfied if $q=\tilde{p}^{(k)}$.
\end{proof}

Following Ref.~\cite{amari2001information}, let us consider the mixed coordinates $\nu_k =(\eta_{k^{-}};\xi_{k^{+}})$. Then, due to the duality of $\eta$ and $\xi$, one can verify that $\tilde{p}^{(k)}$ satisfy the mixed coordinates $\tilde{\nu}_k = (\eta_{k^{-}};0)$, where $\eta_{k^{-}}$ are the constraints of order up to $k$ of $p$. Interestingly, note that $u = \tE_0(0)$ is equal to the uniform distribution $u$ for all $p\in\mathcal{M}$ and all $\gamma$.

\begin{figure}
  \includegraphics[width=8.5cm,height=6cm,keepaspectratio]{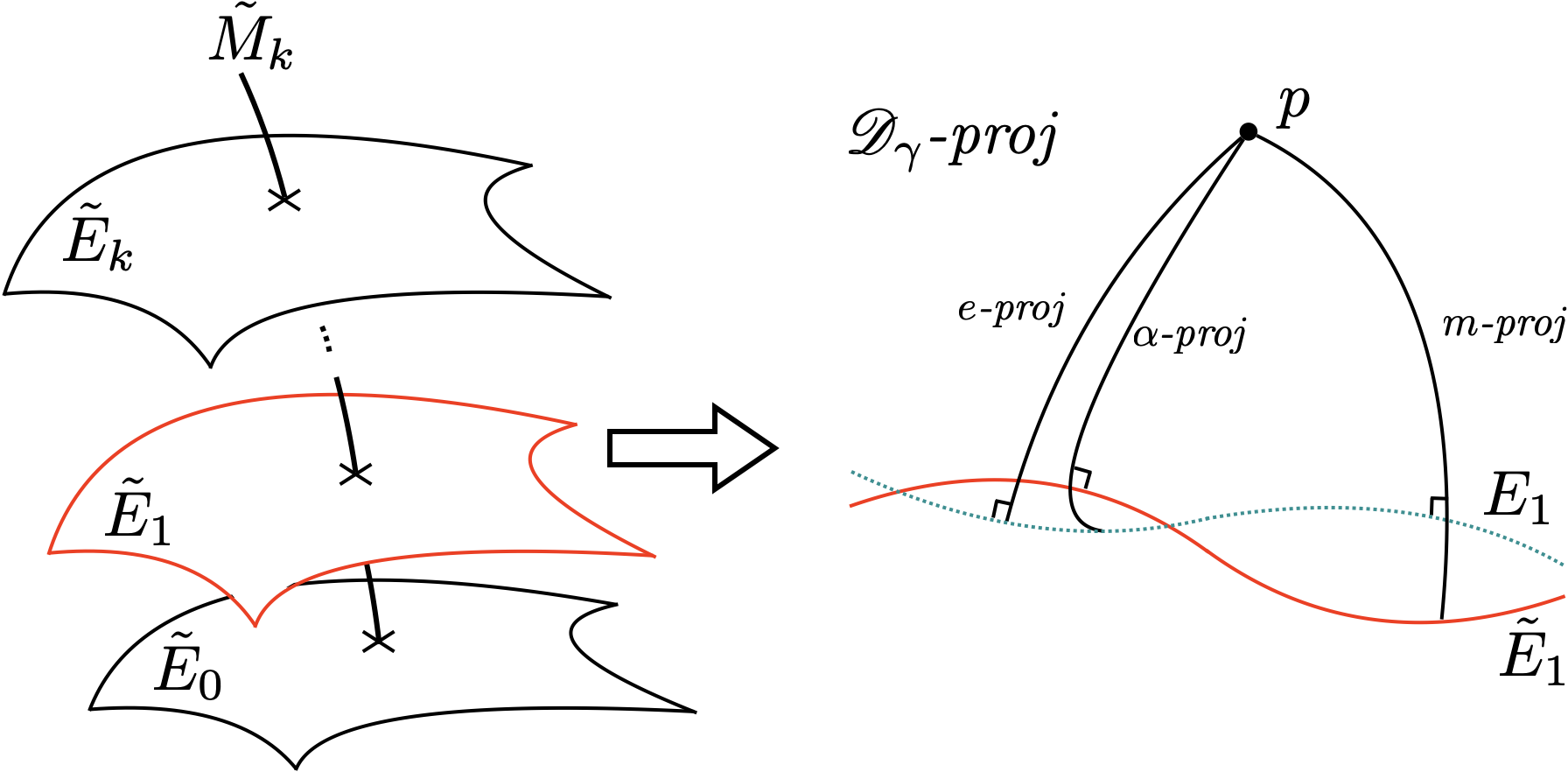}
  \caption{\textit{(left)} Orthogonal foliation of manifold $\mathcal{M}$. \textit{(right)} Projections onto $\textbf{E}_1$ leaf (associated with $\alpha=1$) and its deformation $\tE_1$ related to $\alpha\neq\pm1$.}
  \label{fig:fol_struct}
\end{figure}

With these definitions at hand, we can prove the following result.

\begin{theorem}\label{th:main}
For a given $p\in\mathcal{M}$, the collection of the $\gamma$-projections $\tilde{p}^{(n-1)},\dots,u$ satisfy
\begin{equation}
\label{eq:div_decomp}
  \mathcal{D}_{\gamma}(p||u) = \sum_{k=1}^{n}\mathcal{D}_{\gamma}(\tilde{p}^{(k)}||\tilde{p}^{(k-1)})~.
\end{equation}
\end{theorem}

\begin{proof}
Let's start noting that both $\tilde{p}^{(n-1)}$ and $u$ belong to $\tE_{n-1}$, while 
both $p$ and $\tilde{p}^{(n-1)}$ belong to $\tM_{n-1}$ due to Lemma~\ref{lemma:inc}.  
Therefore, Lemma~\ref{lemma:orto} implies that
\begin{equation}
\label{eq:cool}
  \mathcal{D}_{\gamma}(p||u) = \mathcal{D}_{\gamma}(p||\tilde{p}^{(n-1)})
  +\mathcal{D}_{\gamma}(\tilde{p}^{(n-1)}||u)~.
\end{equation}
The rest of the proof can be done following a similar rationale recursively on $\mathcal{D}_{\gamma}(\tilde{p}^{(n-1)}||u)$.

\end{proof}

To better understand the deformation of the layers induced by $\gamma$, it is beneficial to consider the mean-field theory approach presented in Ref.~\cite{amari2001information2}. Let's consider a classic Ising model for which two layers suffice to describe the system, and focus in its projection to $E_1$.  In~\cite{amari2001information2} the $m$ and $e$ projections denote the solution and naive approximations, respectively, which are both orthogonal. Moreover, the $\alpha$-projection draws the trajectory of solutions in between. In the current picture, however, the submanifolds are deformed in such a way that the $\alpha$-projection becomes orthogonal with $\alpha=\pm 1$, which are left as fixed points (see Figure~\ref{fig:fol_struct}).

\section{Generalising the Maximum Entropy principle}
\label{sec:maxentRenyi}

\subsection{R\'{e}nyi's entropy and related quantities}

Consider a manifold of distributions whose support allows a flat distribution. Then, the $\alpha$-\emph{negentropy} of $p$ is defined as
\begin{equation}
    \pazocal{N}_{\gamma}(p) 
    := \Lambda - H_{\gamma}(p)~, 
\end{equation}
with $H_\gamma =\Lambda$ being the R\'enyi entropy of the uniform distributions, which corresponds to $\log |\bm \chi|$ for finite $\bm \chi$ or $\log n$ in the continuum, and 
\begin{equation}
  H_{\gamma}(p) = \frac{-1}{\gamma}\log \int_{\boldsymbol\chi} p(x;\xi)^{\gamma+1} d\mu(x)
\end{equation}
being the well-known R\'enyi entropy. This definition recovers the standard Shannon entropy and negentropy in the case $\gamma=0$~\footnote{One can interpret a continuous decrease in the constant sectional $\alpha$-curvature of the manifold manifesting as a decrease in the order of R\'enyi's entropy in statistics, eventually converging to Shannon's for $\gamma \to 0$ limit.}.

Another quantity of interest is the $\gamma$-\emph{Total Correlation}, defined as
\begin{align}
  \textrm{TC}_{\gamma}(\bm X^n)
  =
  \sum_{i=0}^n H_{\gamma}(X^i) - H_{\gamma}
  (\bm X^n),
\end{align}
where $\bm X^n := (X_1,\dots,X_n)$ is a random vector that distributes according to $p_\xi(X=x)$ with $x=(x_1,\dots,x_n)$. 
This is a generalization of the well-known Total Correlation for Shannon's entropy (also known as Multi-information~\cite{studeny1998multiinformation}), which is a generalization of Shannon's mutual information for the case of 3 or more variables~\cite{rosas2019quantifying}. In particular, if $n=2$ then the total correlation gives a R\'enyi's mutual information.

\subsection{A hierarchical decomposition of R\'{e}nyi's entropy}

With a hierarchical decomposition $p,p^{(n-1)},\dots,u$ at hand, we are now poized to address the problem of entropy decomposition based on the relevance of each order. 

\begin{lemma}\label{lemma:telescopic}
Consider a the $\gamma$-projections of $p\in\mathcal{M}$ under an orthogonal foliation $\{\tM_k\{p\},\tE_k\}$ such that $\tE_0 = \{u\}$ with $u$ the uniform distribution. 
Then, the following holds for $l<k$:

\begin{equation}
    \mathcal{D}_{\gamma}(\tilde{p}^{(k)}||\tilde{p}^{(l)}) = H_{\gamma}(\tilde{p}^{(l)}) - H_{\gamma}(\tilde{p}^{(k)})~.
\end{equation}
\end{lemma}
\begin{proof}
A direct application of Eq.\eqref{eq:div_decomp} shows that
\begin{equation}
\mathcal{D}_{\gamma}(\tilde{p}^{(k)}||u)
= \mathcal{D}_{\gamma}(\tilde{p}^{(k)}||\tilde{p}^{(l)}) 
+
\mathcal{D}_{\gamma}(\tilde{p}^{(l)}||u)~.    
\end{equation}
Then, the desired result follows from re-ordering the terms and using the fact that $\mathcal{D}_{\gamma}(q||u) = \Lambda - H_{\gamma}(q)$ for any $q\in\mathcal{M}$.
\end{proof}

\begin{corollary}\label{corr:simple}
For any multivariate distribution $p$ then
\begin{align}
    \pazocal{N}_{\gamma}(p) 
    &= \mathcal{D}_{\gamma}(p||u)~, \\ 
    \mathrm{TC}_{\gamma}(\bm X^n)
    &=\mathcal{D}_{\gamma}\left( p\Big|\Big|\prod_{k=1}^n p_{X_k} \right)~. 
\end{align}
\end{corollary}

Using this lemma, we can put forward our main result.

\begin{theorem}\label{teo:2}
Consider $p\in\mathcal{M}$ and an orthogonal foliation $\{\tM_k\{p\},\tE_k\}$ such that $\tE_0 = \{u\}$. Then,
\begin{equation}
    \tilde{p}^{(k)} = \argmax_{q\in \tM_k\{p\}} H_{\gamma}(q)~.
\end{equation}
Additionally, the R\'enyi negentropy can be decomposed as
\begin{equation}\label{eq:negent_dec}
    \pazocal{N}_{\gamma}(p) = \sum_{k=1}^N \Delta^{(k)} H_{\gamma}(p)~,
\end{equation}
with $\Delta^{(k)} H_{\gamma}(p):= H_{\gamma}\big(\tilde{p}^{(k-1)}\big)
- 
H_{\gamma}\big(\tilde{p}^{(k)}\big) >0$ quantifying the relevance of the $k$-th order constraints.
\end{theorem}

\begin{proof}
Because $\tilde{p}^{(k)}\in\tM_k$ (see Lemma~\ref{lemma:inc}), then thanks to Lemma~\ref{lemma:orto} any $r\in \tM_k$ satisfies 
\begin{equation}
\mathcal{D}_{\gamma}(r||u)
=
\mathcal{D}_{\gamma}(r||\tilde{p}^{(k)})
+ 
\mathcal{D}_{\gamma}(\tilde{p}^{(k)}||u)~.
\end{equation}
Therefore, $\mathcal{D}_{\gamma}(r||u)\geq \mathcal{D}_{\gamma}(\tilde{p}^{(k)}||u)$ for all $r\in\tM_k$, and hence it follows that
\begin{equation}
    \tilde{p}^{(k)}
    =
    \argmin_{q\in\tM_k}
    \mathcal{D}_{\gamma}(q||u)
    = \argmax_{q\in\tM_k}
    H_{\gamma}(q)~.
\end{equation}
Above, the first equality is due to the fact that $\tilde{p}^{(k)}\in\tM_k$, and the second equality uses the fact that $\mathcal{D}_{\gamma}(q||u) = \Lambda - H_{\gamma}\big(q\big)$.

To prove Eq.~\eqref{eq:negent_dec}, one can use Corollary~\ref{corr:simple} and Theorem~\ref{th:main} to show that
\begin{equation}
    \pazocal{N}_{\gamma}(p) =
    \mathcal{D}_{\gamma}(p||u)
    =\sum_{k=1}^N \mathcal{D}_{\gamma}(\tilde{p}^{(k)}||\tilde{p}^{(k-1)})~.
\end{equation}
The desired result is then a consequence of Lemma~\ref{lemma:telescopic}.
\end{proof}

Above, $\Delta^{(k)} H_{\gamma}(p)$ accounts for the relevance of the $k$-th order interactions. In particular, the first order term accounts for all the non-interactive part:
\begin{equation}
    \Delta^{(1)} H_{\gamma}(p) 
    = \sum_{j=1}^N
     \pazocal{N}_{\gamma}(X_j) 
     = \sum_{j=1}^N 
     \Big( \log n - H_{\gamma}\big(X_j\big)
     \Big)
\end{equation}
with $\pazocal{N}_{\gamma}(X_j)$ being the marginal negentropy of $X_j$. The remaining terms can be seen to be equal to
\begin{equation}
    \sum_{k=2}^N \Delta^{(k)} H_{\gamma}(p) 
    = \text{TC}_{\gamma}(p)
\end{equation}
showing that the $\text{TC}_\gamma$ captures all the correlated part of the R\'enyi negentropy, following the relationship observed in Shannon's case for $\gamma=0$~(as discussed in Ref.\cite{rosas2019quantifying}).

\subsection{Maximum R\'enyi entropy distributions over constraints on average observables}

Let us now consider a collection of observables $h$ over a system of $n$ binary variables defined as 
\begin{equation}
  h^{i,k}(x) = \prod_{j=1}^k x_{I_i^k(j)}~,
\end{equation}
with $h^{i,k}$ being the $i$-th observable of $k$-th order, with $I_i^k(j)$ being an appropriate assignment of indices. Then, one can define the following coordinates:
\begin{equation}
    \nu^{i,k} := \mathbb{E}\{h^{i,k}(x)\}~.
\end{equation}
For example, $\nu^{i,1}$ are of the form $\mathbb{E}\{x_i\}$ and $\nu^{j,2}$ of the form $\mathbb{E}\{x_r x_s\}$. Importantly, given that $x_1,\dots,x_n$ are binary variable then one can check that, once $\nu^{i,l}$ for all $i$ and $l\leq k$ are fixed, this in turn determines all the $k$-th order marginals~\footnote{The $k$-th order marginals of $p$ are the distributions considering $k$ of the $n$ variables that compose $x$, which are obtained by marginalising the other $n-k$ variables.}. Crucially, this implies that the parameters $\nu$ as a whole determine a unique distribution $p_\nu(x)$, and hence $\nu$ is a valid parametrization of the corresponding statistical manifold~\cite{bialek2012statistical,rosas2016understanding}.

Let us now consider the family of sets $\tM_k$, as defined in Eq.\eqref{eq:M} associated to this parametrization. According to the previous discussion, $\tM_k\{p\}$ is the set of all distributions for $x$ that are compatible with the $k$-th order marginals. For determining the form of the corresponding $k$-th order $\gamma$-projection, we use the following lemma.

\begin{lemma}\label{lemma:asdggh}
The solution of the optimization problem
\begin{equation}\label{eq:optimis}
    \argmax_{q\in\mathcal{M}} H_{\gamma}(q)
    \qquad \text{s.t.}
    \quad
    \nu^{i,l} = \mathbb{E}_q\{h^{i,l}(x)\}
\end{equation}
for all $i$ and $l\leq k$ gives a projection of the form
\begin{equation}\label{eq:maxent_dist}
    \tilde{p}_\theta^{(k)}(x) = e^{-z_\gamma(\theta)} \big(1 + \gamma \theta \cdot h(x) \big)^{1/\gamma} ~,
\end{equation}
with $\theta^{i,l}=0$ for all $l>k$, and a normalization factor given by $z_\gamma(\theta) = \frac{1}{\gamma} \log \sum_{x} \big(1 + \gamma\theta\cdot h(x) \big)^{1/\gamma}$.
\end{lemma}
\begin{proof}
Using Theorem~\ref{teo:2}, it is clear that $\tilde{p}_\theta^{(k)}$ can be found by solving the extreme values of a Lagrangean of the form
\begin{align}
    L(q, \theta_0, \{\theta_j\}) =& H_{\gamma}(q) + \theta_0 \Big(\sum_i q_i -1\Big) \nonumber \\
    &+ \sum_j \theta_j \Big( \sum_k q_k F_j(x_k) - \nu_j \Big)~,
\end{align}
where $q$ is a discrete distribution and $\theta_j$ are Lagrange multipliers. 
The desired result follows from imposing $\partial L/\partial q_i = 0$ and $\partial L/\partial \theta_j = 0$.
\end{proof}
Efficient numerical methods to estimate distributions of the form specified by Eq.~\eqref{eq:maxent_dist} will be developed in a separate publication.

\section{Conclusion}\label{sec:discussion}

This paper shows how the non-Euclidean geometry of curved statistical manifolds naturally leads to a MEP that uses the R\'enyi entropy, generalising the traditional MEP framework based on Shannon's --- which take place on flat manifolds. This generalization of the MEP has three important consequences:
\begin{itemize}
    \item It highlights special geometrical properties of the R\'enyi entropy, which make it stand apart from other generalized entropies.
    \item It provides a solid mathematical foundation for the numerous applications of the R\'enyi entropy and divergence.
    \item It enables a range of novel methods of analysis for the statistics of complex systems.
\end{itemize}

R\'{e}nyi's entropy and divergence represent one of many routes by which the classic information-theoretic definitions can be extended. One fundamental feature of the R\'enyi divergence --- that this work thoroughly exploits --- is the correspondence that it establishes between orthogonality with respect to Fisher's metric and a Pythagorean relationship in the divergence (which does not hold in the geometry induced by e.g. the $\alpha$-divergence). This correspondence is the key property that allows us to build hierarchical foliations, despite the fact that in curved manifolds the link between geometric and Fenchel-Legendre duality is generally broken. It is relevant to highlight that the correspondence between orthogonality and the Pythagorean relationship is not guaranteed by other divergences such as the $\alpha$-divergence, which makes entropies such as Tsallis'~\footnote{For an explanation of the close relationship between the Tsallis entropy and the  $\alpha$-divergence, please see~\cite{Tsallis-alphadiv}.} not well suited to extend the MEP --- at least from an information geometry perspective~\footnote{For an interesting related discussion, including thermodynamic aspects, see Ref.~\cite{scarfone2020study}}. 
Considering that extensions of the Renyi entropy exist (e.g. Ref.~\cite{de2016geometry}), 
an interesting open question is to determine the range of divergences that satisfy these properties.

These findings are in agreement with recent research that is revealing special features of the R\'enyi entropy and divergence in the context of statistical inference and learning. In particular, Refs.~\cite{esposito2019generalization,esposito2020robust} show that the R\'enyi divergence can provide bounds to the generalization error of supervized learning algorithms. Also, Ref.~\cite{jizba2019maximum} shows that the R\'enyi entropy belongs to a class of functionals that are particularly well-suited for inference and estimation. Put together, these findings suggest that the R\'enyi entropy and divergence might be capable of playing an important role in the development of future data analysis and artificial intelligence frameworks.

This work opens the door to novel data-analyses approaches to study high-order interactions. While commonly neglected, high-order statistics have recently been proven to be instrumental in a wide range of phenomena at the heart of complex systems, including the self-organising capabilities of cellular automata~\citep{rosas2018information}, gene-to-gene information flow~\cite{cang2020inferring}, neural information processing~\citep{wibral2017partial}, high-order brain functions~\citep{luppi2020synergistic1,luppi2020synergistic2}, and emergent phenomena~\citep{rosas2020reconciling,varley2021emergence}. However, exhaustive modeling of high-order effects requires an exponential number of parameters; for that reason, practical investigations need to rely on heuristic modeling methods (see e.g.~\cite{ganmor2011sparse,shimazaki2015simultaneous}). In contrast, our framework allow us to do projections while optimising the manifold's curvature in order to best match empirical statistics. Importantly, $k$-th order projections on curved spaces lead to  distributions that capture statistical phenomena of order higher than $k$ without increasing the dimensionality of the parametric family. The development of this line of research is part of our future work.

Another set of promising applications is found in condensed matter systems, where the R\'enyi entropy is often introduced as a measure of the degree of quantum entanglement. In particular, the R\'enyi entropy results from an heuristic generalization of the Von Neumann entropy, which has important benefits in being (i) more suitable to numerical simulations~\cite{hastings2010measuring} and (ii) being easier to measure by experiments~\cite{islam2015measuring}. In particular, the R\'enyi entropy has been shown to be sensible to features of quantum systems such as central charge~\cite{GeometricMutInf}, and knowledge of it at all orders encodes the whole entanglement spectrum of a quantum state~\cite{ShannonRenyiQuantumSpin}.
Moreover, in strongly coupled systems, R\'enyi entropies have been essential for establishing a connection between quantum entanglement and gravity~\cite{Dong:2016fnf,Barrella:2013wja}. More recently, the R\'enyi mutual information has been taking a central role in the identification of phase  transitions~\cite{GeometricMutInf,DetectingPhaseTwithRenyi,PhysRevLett.107.020402}. 
The mathematical framework established in this work serves as a solid basis for these investigations, and further allows the exploration of novel application of information geometry tools in these scenarios.

It is our hope that this contribution may serve to widen the range of applicability of the MEP, while fostering theoretical and practical investigations related to the properties of curved statistical manifolds.

\begin{acknowledgments}
The authors thank Shunichi Amari for careful reading of the manuscript and a number of insightful suggestions, and Ryota Kanai and Yike Guo for supporting this research.
F.E.R. is supported by the Ad Astra Chandaria foundation.
\end{acknowledgments}

\appendix

\section{Deformed exponential family distributions} 
\label{app:def_exp_dist}

For completeness, this appendix presents a derivation of the functional form of $\tilde{p}_\xi$ as presented by Eq.~\eqref{eq:def_exp} that follows Ref.~\cite[Sec.~4.1]{wong2018logarithmic}. 
For this, let us consider an ``exponentially-flat'' manifold~\footnote{A direct calculation shows that parametrizations based on exponential family distributions generate a flat connection.}, i.e. a manifold $\mathcal{M}$ with a parametrization $\xi$ such that all $p\in\mathcal{M}$ can be expressed as
\begin{equation}\label{eq:exp_fam}
  p_\xi(x) = e^{-\xi \cdot h(x) + \phi (\xi)},
\end{equation}
where $h(x)$ is a vector of sufficient statistics of $x$, and $-\phi(\xi)$ is the cumulant generating function. Note that this ``natural parametrization'' of $\mathcal{M}$ allows to express the corresponding contrast function of the KL, $\pazocal{D}_\text{KL}[\xi;\xi']:=\pazocal{D}_\text{KL}(p_\xi||p_{\xi'})$, as a Bregman divergence:
\begin{equation}
    \pazocal{D}_\text{KL}[\xi;\xi']
    =
    (\xi-\xi')\cdot\eta - \phi(\xi) + \phi(\xi')~.
\end{equation}

To find a ``deformed'' exponential distribution $\tilde{p}\in\mathcal{M}$, 
one needs to find the natural parametrization of $\mathcal{M}$ that allows to express the R\'enyi entropy as a Bregman-like divergence. For this purpose, 
one can rewrite Eq.~\eqref{eq:exp_fam} in its self-dual form to find
\begin{equation}\label{eq:KL_exponential}
  \log p_\xi(x) = -\pazocal{D}_\text{KL}[\xi:\xi'] - \psi \big(h(x)\big),
\end{equation}
with $\psi$ the conjugate of $\phi$, and $h(x)$ plays the role of the dual variable $\eta'$. Then, one can re-write Eq.~\eqref{eq:KL_exponential} replacing $\pazocal{D}_\text{KL}$ with $\mathcal{D}_{\gamma}$, and use Eq.~\eqref{eq:def_bigd} to obtain
\begin{align}\label{eq:exp_def_f}
  \log \tilde{p}_\xi(x) 
  &= -\mathcal{D}_{\gamma}[\xi:\xi'] - \psi_{\gamma} (h(x)) \\
  &= -\frac{1}{\gamma} \log\big(1 + \gamma \xi\cdot h(x)\big) + \varphi_{\gamma} (\xi)~,
\end{align}
which leads to
\begin{equation}
  \tilde{p}_\xi(x) = \big(1+ \gamma \xi \cdot h(x)\big)^{-\frac{1}{\gamma}} e^{\varphi_{\gamma}(\xi)}
\end{equation}
with a normalising potential given by Eq.~\eqref{eq:def_varphi}. Importantly, one can show that~\cite[Th.13]{wong2018logarithmic}
\begin{equation}
    \mathcal{D}_\gamma(\tilde{p}_\xi || \tilde{p}_{\xi'})
    = \mathcal{D}_\gamma[\xi;\xi']~,
\end{equation}
which confirms that the parametrization of $\mathcal{M}$ determined by Eq.~\eqref{eq:def_exp} is the natural (in the Bregman-like sense) parametrization of the deformed geometry induced by $\mathcal{D}_\gamma$.

\section{Analysis of deformed expectation values}
\label{app:special_cases}

The deformed expectation values given by Eq.~\eqref{eq:dual_esp} are non-trivial to interpret, and their explicit dependence on $\xi$ makes numerical simulation challenging. However, exploring some ranges of values of $\gamma$ can help us to flesh out an interpretation for $\eta$. 

To this end, let us start by considering the Taylor series expansion of the $Z_\xi$ field given by
\begin{equation}
    Z_{\xi}(h) = h(X)\sum_{n=0}^{\infty} (-1)^{n}(\gamma \xi \cdot h(X))^{n}~.
\end{equation}
Small values of $\gamma$ ensure convergence of the series. Now, one may write its expectation value as
\begin{align}
    \mathbb{E}_{\xi}\{Z_{\xi}^{i}(h)\}  \simeq &\;\mathbb{E}_{\xi}\{ h^{i} \} - \gamma \xi^j  \mathbb{E}_{\xi}\{h^i h_j \} \nonumber \\ & + \gamma^2 \xi^j\xi^k \mathbb{E}_{\xi}\{h^{i}h_{j}h_{k} \}~,
\end{align}
where we have retained up to second order corrections. Similarly for $\eta$, one can find that
\begin{align}
\label{eq:expansion}
    \eta_i & \simeq \mathbb{E}_{\xi}\{h_i\} - \xi^j (\mathbb{E}_{\xi}\{h_i h_j\} +\mathbb{E}_{\xi}\{h_j\}\mathbb{E}_{\xi}\{h_i\})\gamma \nonumber \\
    & \quad + \xi^j \xi^k (\mathbb{E}_{\xi}\{h_i h_j h_k\} + \mathbb{E}_{\xi}\{h_j h_k \}\mathbb{E}_{\xi}\{h_i\} \nonumber \\
    & \qquad \qquad + \mathbb{E}_{\xi}\{h_j \}\mathbb{E}_{\xi}\{h_k h_i\} \nonumber \\
    & \qquad \qquad + \mathbb{E}_{\xi}\{ h_j\} \mathbb{E}_{\xi}\{h_k\}\mathbb{E}_{\xi}\{h_i\})\gamma^2~.
\end{align}
This implies that these Bregman-like dual coordinate generally deviates from the one obtained for $\gamma=0$ through higher orders moments, which becomes more prominent as one increases the order of its $\gamma$-expansion.

\section{Pythagorean relation}
\label{app:pythagorean}

This appendix provides a proof for Lemma~\ref{lemma:orto}, which follows results presented in Ref.~\cite{wong2018logarithmic}.
\begin{proof}
Let's consider a primal geodesic connecting $p$ and $q$ with coordinates $\xi$ and a dual geodesic connecting $r$ and $q$ with coordinates $\eta$. The geodesics are then proportional to $\xi^i_r-\xi^i_q$ and $\eta_{p,j}-\eta_{q,j}$ respectively. Then, let's define
\begin{align}\label{eq:AB}
    A &= \sum_i (\xi^i_r-\xi^i_q)\partial_{\xi^i}~, \\
    B &= \sum_j (\eta_{p,j}-\eta_{q,j})\partial_{\eta_j}~,
\end{align}
and take a look of their inner product
\begin{align}
    \langle A,B \rangle &= \Big \langle \sum_{i}(\xi_r^i-\xi_q^i)\partial_{\xi^i}, \sum_{j}(\eta_q^j-\eta_r^j)\partial_{\eta_j} \Big \rangle \\
    &= \sum_{i,j}(\xi_{r}^{i} - \xi_{q}^{i})(\eta_{p,j} - \eta_{q,j}) \langle \partial_{\xi^i},\partial_{\eta_j}\rangle~.\label{eq:innerprodAB}
\end{align}
In other words, we rely on the evaluation of~\eqref{eq:innerprodAB}, which requires that the inner product of the primal and dual bases induced by the divergence~\eqref{eq:Renyi_div}, vanish. That is,
\begin{align}
    \langle \partial_{\xi^i},\partial_{\eta_j} \rangle & = \Big \langle \partial_{\xi^i},\sum_m \partial_{\eta_j}\xi^m \partial_{\xi^m} \Big \rangle \\ \label{eq:r_metric}
    & = \sum_{m} \partial_{\eta_j}\xi^m  \langle \partial_{\xi^i},\partial_{\xi^m} \rangle \vphantom{\int}~,
\end{align}
whose intern product can be directly obtained from the divergence as
\begin{align}
    & \tilde{g}_{im}(\xi) = -\partial_i \partial_{m'} \mathcal{D}_{\gamma}[\xi,\xi']|_{\xi' =\xi} \\
    & = \left. \left \{ \frac{-\partial_{\xi^{\prime m}}\eta'_{i}}{\Pi (\xi,\eta')} + \sum_{l} \frac{\gamma\, \eta'_{i}\xi^{l}}{\Pi (\xi,\eta')^2}\partial_{\xi^{\prime m}}\eta'_{l} \right \} \right|_{\xi' = \xi}~,
\end{align}
where we use the shorthand notation $\Pi(\xi,\eta'):= (1+ \gamma \xi \cdot \eta')$. Replacing this expression into~\eqref{eq:r_metric} yields
\begin{equation} \label{eq:inner_prod}
    \langle \partial_{\xi^i},\partial_{\eta_j} \rangle = \frac{-1}{\Pi (\xi,\eta)}\delta_{i}^j +\frac{\alpha}{\Pi (\xi,\eta)^2}\eta_{i}\xi^{j}~.
\end{equation}
Using this in Eq.~\eqref{eq:innerprodAB}, and adopting $\Pi_q :=\Pi(\xi_q,\eta_q)$ for brevity, one finds that
\begin{equation}
    \langle A,B \rangle = \sum_{i,j} (\xi_{r}^{i} - \xi_{q}^{i})(\eta_{p,j} - \eta_{q,j}) \left( \frac{-1}{\Pi_{q}}\delta_{i}^j -\frac{\alpha}{\Pi_{q}^2}\eta_{q,i}\xi_{q}^{j} \right)~.
\end{equation}
Evaluating the sum, one finds that this expression is proportional to
\begin{equation} \label{eq:condition}
    \Pi_{q} (\xi_{r} - \xi_{q}) \cdot (\eta_{p} - \eta_{q}) 
    + \alpha \xi_{q} \cdot (\eta_{p} - \eta_{q})\eta_q \cdot (\xi_{r} - \xi_{q})
\end{equation}
Finally, the Pythagorean relationship in Eq.~\eqref{eq:Pyth_rel} holds
\begin{align}
    \iff (1 + \gamma \xi_{q} \cdot \eta_{p})(1 + \gamma \xi_r \cdot \eta_q) & = \nonumber \\
    (1 + \gamma \xi_r \cdot \eta_p)& (1 + \gamma \xi_q \cdot \eta_q) \\
    \iff (\xi_{r} - \xi_{q}) \cdot (\eta_{p} - \eta_{q}) & = \nonumber \\ \label{eq:pyth_cond}
    \gamma (\xi_q \cdot \eta_p) (\xi_r \cdot \eta_q) - \gamma & (\xi_r \cdot \eta_p)(\xi_q \cdot \eta_q)
\end{align}
as it can be seen directly from its logarithmic dependence and the Fenchel-Lengendre relation for the scalar potentials on point $q$. Since the primal geodesic and its dual are orthogonal at $q$, this~\eqref{eq:condition} must vanish resulting in~\eqref{eq:pyth_cond}, hence the Pythagorean relation holds.
\end{proof}

\bibliography{references}

\begin{thebibliography}{99}
\bibitem [{\citenamefont {Amari}()}]{Amaridivergences}%
  \BibitemOpen
  \bibfield  {author} {\bibinfo {author} {\bibfnamefont {Shun-Ichi}\
  \bibnamefont {Amari}},\ }\bibfield  {title} {\enquote {\bibinfo {title}
  {$\alpha$-divergence is unique, belonging to both $f$-divergence and bregman
  divergence classes},}\ }\href@noop {} {\ }\BibitemShut {NoStop}%
\end{thebibliography}
\end{document}